\documentclass[12pt,oneside]{amsart}
\usepackage{amssymb}
\usepackage{amsfonts}
\usepackage{amscd}
\usepackage[matrix, arrow, curve]{xy}
\usepackage{graphicx}
\numberwithin{equation}{section}
\newtheorem{theorem}{Theorem}[section]

\theoremstyle{definition}
\newtheorem{definition}[theorem]{Definition}
\theoremstyle{remark}
\newtheorem*{remark}{Remark}

\newtheorem*{conjecture}{Conjecture}

\begin{document}
%%%%%%%%%%%%%%%%%%%%%%%%%%%%%%%%%%%%%%%%%%%%%%%%%
%%%%%%%%%%%%  macrodefinitions
%%%%%%%%%%%%%%%%%%%%%%%%%%%%%%%%%%%%%%%%%%%%%%%%%
%  Macros (general)
%%%%%%%%%%%%%%%%%%%%%%%
%\newcommand{\MgNekp}{\mathcal{M}_{g,N+1}^{(k,p)}} %% moduli space
%\newcommand{\M}{\mathcal{M}_{g,N+1}^{(1)}}
\newcommand{\M}{\mathcal{M}}
\newcommand{\F}{\mathcal{F}}

\newcommand{\Teich}{\mathcal{T}_{g,N+1}^{(1)}}
\newcommand{\T}{\mathrm{T}}
%%%%   temporary
\newcommand{\corr}{\bf}
\newcommand{\vac}{|0\rangle}
\newcommand{\Ga}{\Gamma}
\newcommand{\new}{\bf}
\newcommand{\define}{\def}
\newcommand{\redefine}{\def}
\newcommand{\Cal}[1]{\mathcal{#1}}
\renewcommand{\frak}[1]{\mathfrak{{#1}}}
\newcommand{\Hom}{\rm{Hom}\,}
%%%%%%%%%%%%%%%%%%%%%%%%%%%%%%%%%%%%
%   Referencing Scheme of Martin
%%%%%%%%%%%%%%%%%%%%%%%%%%
\newcommand{\refE}[1]{(\ref{E:#1})}
\newcommand{\refCh}[1]{Chapter~\ref{Ch:#1}}
\newcommand{\refS}[1]{Section~\ref{S:#1}}
\newcommand{\refSS}[1]{Section~\ref{SS:#1}}
\newcommand{\refT}[1]{Theorem~\ref{T:#1}}
\newcommand{\refO}[1]{Observation~\ref{O:#1}}
\newcommand{\refP}[1]{Proposition~\ref{P:#1}}
\newcommand{\refD}[1]{Definition~\ref{D:#1}}
\newcommand{\refC}[1]{Corollary~\ref{C:#1}}
\newcommand{\refL}[1]{Lemma~\ref{L:#1}}
\newcommand{\refEx}[1]{Example~\ref{Ex:#1}}
%%%%%%%%%%%%%%%%%%%%%%%%%%%%%%%%%%
\newcommand{\R}{\ensuremath{\mathbb{R}}}
\newcommand{\C}{\ensuremath{\mathbb{C}}}
\newcommand{\N}{\ensuremath{\mathbb{N}}}
\newcommand{\Q}{\ensuremath{\mathbb{Q}}}
\renewcommand{\P}{\ensuremath{\mathcal{P}}}
\newcommand{\Z}{\ensuremath{\mathbb{Z}}}
%%%%%%%%%%%%%%%%%%%%%%%%%%%%%%%%%%%%%%%%%%
\newcommand{\kv}{{k^{\vee}}}
%%%%%%%%%%%%%%%%%%%%%%%%%%%%%%%%%%%%%%%%%%%%%
\renewcommand{\l}{\lambda}
%%%%%%%%%%%%%%%%%%%%%%%%%%%%%%%%%%%%%%%%%%%%%%%%%%
\newcommand{\gb}{\overline{\mathfrak{g}}}
\newcommand{\dt}{\tilde d}     % Oleg
\newcommand{\hb}{\overline{\mathfrak{h}}}
\newcommand{\g}{\mathfrak{g}}
\newcommand{\h}{\mathfrak{h}}
\newcommand{\gh}{\widehat{\mathfrak{g}}}
\newcommand{\ghN}{\widehat{\mathfrak{g}_{(N)}}}
\newcommand{\gbN}{\overline{\mathfrak{g}_{(N)}}}
\newcommand{\tr}{\mathrm{tr}}
\newcommand{\gln}{\mathfrak{gl}(n)}
\newcommand{\son}{\mathfrak{so}(n)}
\newcommand{\spnn}{\mathfrak{sp}(2n)}
\newcommand{\sln}{\mathfrak{sl}}
\newcommand{\sn}{\mathfrak{s}}
\newcommand{\so}{\mathfrak{so}}
\newcommand{\spn}{\mathfrak{sp}}
\newcommand{\tsp}{\mathfrak{tsp}(2n)}
\newcommand{\gl}{\mathfrak{gl}}
\newcommand{\slnb}{{\overline{\mathfrak{sl}}}}
\newcommand{\snb}{{\overline{\mathfrak{s}}}}
\newcommand{\sob}{{\overline{\mathfrak{so}}}}
\newcommand{\spnb}{{\overline{\mathfrak{sp}}}}
\newcommand{\glb}{{\overline{\mathfrak{gl}}}}
\newcommand{\Hwft}{\mathcal{H}_{F,\tau}}
\newcommand{\Hwftm}{\mathcal{H}_{F,\tau}^{(m)}}

%%%%%%%%%%%%%%%%%%%%%%%%%%%%%%%%%%%%%%%%%%%%%%%%%%%%
\newcommand{\car}{{\mathfrak{h}}}    % Cartan subalgebra
\newcommand{\bor}{{\mathfrak{b}}}    % Borel subalgebra
\newcommand{\nil}{{\mathfrak{n}}}    % nilpotent subalgebra
\newcommand{\vp}{{\varphi}}
\newcommand{\bh}{\widehat{\mathfrak{b}}}  % Borel subalgebra of KN algebra
\newcommand{\bb}{\overline{\mathfrak{b}}}  % Borel subalgebra of KN algebra
\newcommand{\Vh}{\widehat{\mathcal V}}
\newcommand{\KZ}{Kniz\-hnik-Zamo\-lod\-chi\-kov}
\newcommand{\TUY}{Tsuchia, Ueno  and Yamada}
\newcommand{\KN} {Kri\-che\-ver-Novi\-kov}
\newcommand{\pN}{\ensuremath{(P_1,P_2,\ldots,P_N)}}
\newcommand{\xN}{\ensuremath{(\xi_1,\xi_2,\ldots,\xi_N)}}
\newcommand{\lN}{\ensuremath{(\lambda_1,\lambda_2,\ldots,\lambda_N)}}
\newcommand{\iN}{\ensuremath{1,\ldots, N}}
\newcommand{\iNf}{\ensuremath{1,\ldots, N,\infty}}

\newcommand{\tb}{\tilde \beta}
\newcommand{\tk}{\tilde \varkappa}
\newcommand{\ka}{\kappa}
\renewcommand{\k}{\varkappa}
\newcommand{\ce}{{c}}

\newcommand{\Pif} {P_{\infty}}
\newcommand{\Pinf} {P_{\infty}}
\newcommand{\PN}{\ensuremath{\{P_1,P_2,\ldots,P_N\}}}
\newcommand{\PNi}{\ensuremath{\{P_1,P_2,\ldots,P_N,P_\infty\}}}
\newcommand{\Fln}[1][n]{F_{#1}^\lambda}
\newcommand{\tang}{\mathrm{T}}
\newcommand{\Kl}[1][\lambda]{\can^{#1}}
\newcommand{\A}{\mathcal{A}}
\newcommand{\U}{\mathcal{U}}
\newcommand{\V}{\mathcal{V}}
\newcommand{\W}{\mathcal{W}}
\renewcommand{\O}{\mathcal{O}}
\newcommand{\Ae}{\widehat{\mathcal{A}}}
\newcommand{\Ah}{\widehat{\mathcal{A}}}
\newcommand{\La}{\mathcal{L}}
\newcommand{\Le}{\widehat{\mathcal{L}}}
\newcommand{\Lh}{\widehat{\mathcal{L}}}
\newcommand{\eh}{\widehat{e}}
\newcommand{\Da}{\mathcal{D}}
\newcommand{\kndual}[2]{\langle #1,#2\rangle}
\newcommand{\cins}{\frac 1{2\pi\mathrm{i}}\int_{C_S}}
\newcommand{\cinsl}{\frac 1{24\pi\mathrm{i}}\int_{C_S}}
\newcommand{\cinc}[1]{\frac 1{2\pi\mathrm{i}}\int_{#1}}
\newcommand{\cintl}[1]{\frac 1{24\pi\mathrm{i}}\int_{#1 }}
\newcommand{\w}{\omega}
\newcommand{\ord}{\operatorname{ord}}
\newcommand{\res}{\operatorname{res}}
\newcommand{\nord}[1]{:\mkern-5mu{#1}\mkern-5mu:}
\newcommand{\codim}{\operatorname{codim}}
\newcommand{\ad}{\operatorname{ad}}
\newcommand{\Ad}{\operatorname{Ad}}
\newcommand{\supp}{\operatorname{supp}}
\newcommand{\semi}{\mathcal{F}^{\infty /2}}

%%%%%%%%%%%%%%%%%%%%%%%%%%%%%%%%%%%%%%%%%%%%%%%%
\newcommand{\Fn}[1][\lambda]{\mathcal{F}^{#1}}
\newcommand{\Fl}[1][\lambda]{\mathcal{F}^{#1}}
\renewcommand{\Re}{\mathrm{Re}}

\newcommand{\ha}{H^\alpha}

\define\ldot{\hskip 1pt.\hskip 1pt}
\define\ifft{\qquad\text{if and only if}\qquad}
\define\a{\alpha}
\redefine\d{\delta}
\define\w{\omega}
\define\ep{\epsilon}
\redefine\b{\beta} \redefine\t{\tau} \redefine\i{{\,\mathrm{i}}\,}
\define\ga{\gamma}
\define\cint #1{\frac 1{2\pi\i}\int_{C_{#1}}}
\define\cintta{\frac 1{2\pi\i}\int_{C_{\tau}}}
\define\cintt{\frac 1{2\pi\i}\oint_{C}}
\define\cinttp{\frac 1{2\pi\i}\int_{C_{\tau'}}}
\define\cinto{\frac 1{2\pi\i}\int_{C_{0}}}
%\define\cinttt{\frac 1{24\pi\i}\int_{C_{\tau}}}
\define\cinttt{\frac 1{24\pi\i}\int_C}
\define\cintd{\frac 1{(2\pi \i)^2}\iint\limits_{C_{\tau}\,C_{\tau'}}}
\define\dintd{\frac 1{(2\pi \i)^2}\iint\limits_{C\,C'}}
\define\cintdr{\frac 1{(2\pi \i)^3}\int_{C_{\tau}}\int_{C_{\tau'}}
\int_{C_{\tau''}}}
\define\im{\operatorname{Im}}
\define\re{\operatorname{Re}}
%\define\res{\text{res}}
\define\res{\operatorname{res}}
\redefine\deg{\operatornamewithlimits{deg}}
\define\ord{\operatorname{ord}}
\define\rank{\operatorname{rank}}
\define\fpz{\frac {d }{dz}}
\define\dzl{\,{dz}^\l}
\define\pfz#1{\frac {d#1}{dz}}

\define\K{\Cal K}
\define\U{\Cal U}
\redefine\O{\Cal O}
\define\He{\text{\rm H}^1}
\redefine\H{{\mathrm{H}}}
\define\Ho{\text{\rm H}^0}
\define\A{\Cal A}
\define\Do{\Cal D^{1}}
\define\Dh{\widehat{\mathcal{D}}^{1}}
\redefine\L{\Cal L}
\newcommand{\ND}{\ensuremath{\mathcal{N}^D}}
\redefine\D{\Cal D^{1}}
\define\KN {Kri\-che\-ver-Novi\-kov}
\define\Pif {{P_{\infty}}}
\define\Uif {{U_{\infty}}}
\define\Uifs {{U_{\infty}^*}}
\define\KM {Kac-Moody}
\define\Fln{\Cal F^\lambda_n}
%%%%%%%%%%%%%%%%%%%%
\define\gb{\overline{\mathfrak{ g}}}
\define\G{\overline{\mathfrak{ g}}}
\define\Gb{\overline{\mathfrak{ g}}}
\redefine\g{\mathfrak{ g}}
\define\Gh{\widehat{\mathfrak{ g}}}
\define\gh{\widehat{\mathfrak{ g}}}
%%%%%%%%%%%%%%%%%%%%%%%%%%
\define\Ah{\widehat{\Cal A}}
\define\Lh{\widehat{\Cal L}}
\define\Ugh{\Cal U(\Gh)}
\define\Xh{\hat X}
\define\Tld{...}
\define\iN{i=1,\ldots,N}
\define\iNi{i=1,\ldots,N,\infty}
\define\pN{p=1,\ldots,N}
\define\pNi{p=1,\ldots,N,\infty}
\define\de{\delta}

\define\kndual#1#2{\langle #1,#2\rangle}
\define \nord #1{:\mkern-5mu{#1}\mkern-5mu:}
%\define \MgN{{\Cal M}_{g,N}} %% moduli space
%\define \MgNp{{\Cal M}_{g,N}^{(p)}} %% moduli space
\newcommand{\MgN}{\mathcal{M}_{g,N}} %% moduli space
\newcommand{\MgNeki}{\mathcal{M}_{g,N+1}^{(k,\infty)}} %% moduli space
\newcommand{\MgNeei}{\mathcal{M}_{g,N+1}^{(1,\infty)}} %% moduli space
\newcommand{\MgNekp}{\mathcal{M}_{g,N+1}^{(k,p)}} %% moduli space
\newcommand{\MgNkp}{\mathcal{M}_{g,N}^{(k,p)}} %% moduli space
\newcommand{\MgNk}{\mathcal{M}_{g,N}^{(k)}} %% moduli space
\newcommand{\MgNekpp}{\mathcal{M}_{g,N+1}^{(k,p')}} %% moduli space
\newcommand{\MgNekkpp}{\mathcal{M}_{g,N+1}^{(k',p')}} %% moduli space
\newcommand{\MgNezp}{\mathcal{M}_{g,N+1}^{(0,p)}} %% moduli space
\newcommand{\MgNeep}{\mathcal{M}_{g,N+1}^{(1,p)}} %% moduli space
\newcommand{\MgNeee}{\mathcal{M}_{g,N+1}^{(1,1)}} %% moduli space
\newcommand{\MgNeez}{\mathcal{M}_{g,N+1}^{(1,0)}} %% moduli space
\newcommand{\MgNezz}{\mathcal{M}_{g,N+1}^{(0,0)}} %% moduli space
\newcommand{\MgNi}{\mathcal{M}_{g,N}^{\infty}} %% moduli space
\newcommand{\MgNe}{\mathcal{M}_{g,N+1}} %% moduli space
\newcommand{\MgNep}{\mathcal{M}_{g,N+1}^{(1)}} %% moduli space
\newcommand{\MgNp}{\mathcal{M}_{g,N}^{(1)}} %% moduli space
\newcommand{\Mgep}{\mathcal{M}_{g,1}^{(p)}} %% moduli space
\newcommand{\MegN}{\mathcal{M}_{g,N+1}^{(1)}} %% moduli space

%\define \mpt{(M,P_1,P_2,\ldots, P_N,\Pif)} %% moduli point
%\define \mpp{(M,P_1,P_2,\ldots, P_N)} %% moduli point
%\define \MgNn{{\Cal M}_{g,N}^{(1)}} %% moduli space
%\define \MgNen{{\Cal M}_{g,N+1}^{(1)}} %% moduli space
%\define \Mgo{{\Cal M}_{g,0}} %% moduli space
%\define \mptn{(M,P_1,P_2,\ldots, P_N,\Pif,z_1,\ldots,z_N,z_\infty)}
 %% moduli point
%\define \mppn{(M,P_1,P_2,\ldots, P_N,z_1,\ldots,z_N)} %% moduli point
\define \sinf{{\widehat{\sigma}}_\infty}
\define\Wt{\widetilde{W}}
\define\St{\widetilde{S}}
\newcommand{\SigmaT}{\widetilde{\Sigma}}
\newcommand{\hT}{\widetilde{\frak h}}
\define\Wn{W^{(1)}}
\define\Wtn{\widetilde{W}^{(1)}}
\define\btn{\tilde b^{(1)}}
\define\bt{\tilde b}
\define\bn{b^{(1)}}
\define \ainf{{\frak a}_\infty} %matrices with a finite number of
                                %diagonals

%
%%%%%%%%%% Olegs definitions %%%%%%%%%%%%%%%%%%%%%%%%%%%%%%%%%%%
\define\eps{\varepsilon}    % Oleg
\newcommand{\e}{\varepsilon}
\define\doint{({\frac 1{2\pi\i}})^2\oint\limits _{C_0}
       \oint\limits _{C_0}}                            % Oleg
\define\noint{ {\frac 1{2\pi\i}} \oint}   % Oleg
\define \fh{{\frak h}}     % Oleg
\define \fg{{\frak g}}     % Oleg
\define \GKN{{\Cal G}}   % affine Krichever-Novikov algebra % Oleg
\define \gaff{{\hat\frak g}}   % affine Krichever-Novikov algebra
\define\V{\Cal V}
\define \ms{{\Cal M}_{g,N}} %% moduli space
\define \mse{{\Cal M}_{g,N+1}} %% moduli space
%%%%%%%%%%%%%%%%%%%%%%%%%%%%%%%%%%%%%%
\define \tOmega{\Tilde\Omega}
\define \tw{\Tilde\omega}
\define \hw{\hat\omega}
\define \s{\sigma}
\define \car{{\frak h}}    % Cartan subalgebra
\define \bor{{\frak b}}    % Borel subalgebra
\define \nil{{\frak n}}    % nilpotent subalgebra
\define \vp{{\varphi}}
\define\bh{\widehat{\frak b}}  % Borel subalgebra of KN algebra
\define\bb{\overline{\frak b}}  % Borel subalgebra of KN algebra
\define\KZ{Knizhnik-Zamolodchikov}
\define\ai{{\alpha(i)}}
\define\ak{{\alpha(k)}}
\define\aj{{\alpha(j)}}
\newcommand{\calF}{{\mathcal F}}
\newcommand{\ferm}{{\mathcal F}^{\infty /2}}
\newcommand{\Aut}{\operatorname{Aut}}
\newcommand{\End}{\operatorname{End}}
%%%%%%%%%%%%%%%%%%%%%%%%%%%%%%%%%%%%%%%%%%%
%%%%%%%%%%%%%%%%%%%%%%%%%%%%%%%%%
%%%%%%%%%%%%%%   for laxcent
%%%%%%%%%%%%%%%%%%%%%%%%%%%%%%%%%%
\newcommand{\laxgl}{\overline{\mathfrak{gl}}}
\newcommand{\laxsl}{\overline{\mathfrak{sl}}}
\newcommand{\laxso}{\overline{\mathfrak{so}}}
\newcommand{\laxsp}{\overline{\mathfrak{sp}}}
\newcommand{\laxs}{\overline{\mathfrak{s}}}
\newcommand{\laxg}{\overline{\frak g}}
\newcommand{\bgl}{\laxgl(n)}
%%%%%%%%%%%%%%%%%%%%%%%%
\newcommand{\tX}{\widetilde{X}}
\newcommand{\tY}{\widetilde{Y}}
\newcommand{\tZ}{\widetilde{Z}}
%%%%%%%%%%%%%%%%%%%%%%%%%%%%%%%%%%%%%%%%%%
%%%%%%%%%%%%  END of macrodefinitions
%%%%%%%%%%%%%%%%%%%%%%%%%%%%%%%%%%%%%%%%%

%%%%%%%%%%%%%%%%%%%%%%%%%%%%%%%%%
%Top-Matter
%%%%%%%%%%%%%%%%%%%%%%%%%%%%%%%
%%%%%%%%%%%%%%%%%    private header  %%%%%%%%%%%%%%%%%%%%

%\large{
\title[]{Quantization of Lax integrable systems \\ and Conformal Field Theory}
\author[O.K.Sheinman]{O.K.Sheinman}
%\date{\today}
\maketitle
\begin{abstract}
We present the correspondence between Lax integrable systems with spectral parameter on a Riemann surface, and Conformal Field Theories, in quite general set-up suggested earlier by the author. This correspondence turns out to give a prequantization of the integrable systems in question.
\end{abstract}
\tableofcontents
\section{Introduction}

%%%%%%%%%%%%%%%%%%%%%%%%%%%%%%%%%%%%%%%%%%%%%%%%%%
We address here the following problem: given a Lax integrable
system with a spectral parameter on a Riemann surface, to construct a unitary
projective representation of the corresponding Lie algebra of
Hamiltonian vector fields. For the Lax equations in question, we
propose a way to represent Hamiltonian vector fields by covariant
derivatives with respect to the Knizhnik--Zamolodchikov
connection. This is a Dirac-type prequantization from the physical
point of view. As well, this establishes a correspondence between Lax integrable systems and Conformal Field Theories.

The idea of quantization of Hitchin systems by
means of the Knizhnik--Zamolodchikov connection was addressed, or at
least mentioned, many times in the theoretical physics literature
(D.Ivanov \cite{D_I}, G.Felder and  Ch.Wieczerkowski \cite{FeW},
M.A.Olshanet\-sky and A.M.Levin \cite{Olsh_Lev, Olsh}) but only
the second order Hamiltonians have been involved. In \cite{Sh_LKZ} such relation was observed for all Hamiltonians, and, moreover, for all classical observables of the Hamiltonian systems given by the Lax equations in question. The observation was based on the theory of Lax integrable systems with spectral parameter on a Riemann surface, originally due to I.Krichever \cite{Klax} and later developed in \cite{KrSh,Sh_DGr,Shein_UMN2016}, and on the global operator approach to quantization of strings \cite{KNFa,KNFb,KNFc,WZWN2,Schlich_DGr,Sh_DGr}. The proof of unitarity in \cite{Sh_LKZ} relied on the Poincar\'e theorem on absolute invariants (see \refS{Unit} below).

Though I~revised \cite{Sh_LKZ} twice \cite{Sh_DGr,Shein_UMN2016}, I was never satisfied with its presentation. In this paper, taking account of recent developments on spectral curves of Hitchin systems \cite{Shein_FAN19,Sh_Bor_19},   and of results by M.Schlichenmaier on almost graded structures on multipoint Krichever--Novikov algebras and modules over them (\cite{Schlich_DGr} and references therein), I do the next attempt. Compared to previous ones, I have restricted myself with the Lax systems whose spectral curves admit holomorphic involutions. It is a particular class, though broad enough (see \refS{CFT} for the discussion). For the beginning it seemed to be a technical restriction, but systematic usage of an involution on the spectral curve finally enabled me to better relate the present subject to the above mentioned Krichever--Novikov works on the operator quantization of string \cite{KNFa,KNFb,KNFc}, and to the results due to Schlichenmaier on almost graded structures \cite{Schlich_DGr}.

Other approaches to quantization of integrable systems (mainly, Hitchin systems) include geometric quantization (N.Hitchin \cite{Flat_conn}, J.E.Andersen), Berezin--Toeplitz and deformation quantizations (J.E.Andersen, M.Schlichenmaier), quantum integrable systems (B.Feigin and E.Frenkel \cite{FF}, A.Beilinson and V.Drinfeld \cite{BD}, A.P.Veselov,
A.N.Se\-r\-ge\-ev, G.Felder, M.V.Feigin). Compared to those approaches, we to a greater extend exploit the Lax nature of the integrable systems in question. The relation of our approach to the one of Hitchin is the same as the relation of the Knizhnik--Zamolodchikov connection to the Hitchin connection. Compared to quantum integrable systems, we prequantize
the whole algebra of observables rather than  any commutative
subalgebra of it.

The present paper is a revised version of the talk by the author at the conference "Homotopy algebras, deformation theory and quantization", S.Banach International Mathematical Center, Bedlewo, Poland, 2018. The author is grateful to the organizers for the invitation and for an excellent opportunity to revisit the subject. I am deeply grateful to the referee of my contribution to Proceedings of the conference for  valuable (though, anonymous) discussion.

%%%%%%%%%%%%%%%%%%%%%%%%%%%%%%%%%%%%%%%%%%%%%%%%%%%%%%%%
\section{Lax integrable systems with a spectral parameter on a Riemann surface}
The long term development of the theory of Lax integrable systems with a finite number of degrees of freedom suggests the following general description of their structure.

We start with a quadruple $\{\Sigma,\Pi,\g,V\}$ where $\Sigma$ is a Riemann surface, $\Pi\subset\Sigma$ is a finite set of marked points, $\g$ is a semi-simple Lie algebra over $\C$, $V$ is its faithful module. Let $G$ be the Chevalley group corresponding to the pair $\{\g,V\}$ \cite{Stein}, $\Pi=\{ P_1,\ldots,P_N \}$.

Given such quadruple, we define first the space of flag configurations. By a \emph{flag configuration}, we mean

1) an ordered set of points $\Gamma=\{\ga_s\in\Sigma\ |\ s=1,\ldots, K\}$ (mutually distinct with $P$-points).

2) the set of corresponding flags $F_s$, $s=1,\ldots, K$ in $V$.

By the \emph{space of flag configurations} we mean the set of all flag configurations with the same $K$, such that $F_s$ and $F_s'$ belong to the same $G$-orbit in the corresponding flag space for every $s=1,\ldots, K$.

We think of the flags $F_s$ as of "located at the points $\ga_s$":
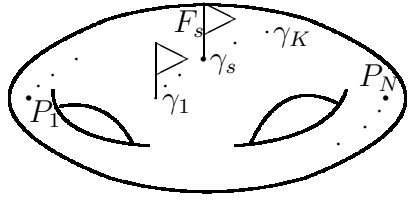
\begin{figure}[h]
\begin{picture}(150,90)
\unitlength=0.6pt
\thicklines
%
% внешний овал
\qbezier(50,20)(-75,70)(50,120) % левая сторона
\qbezier(175,20)(300,70)(175,120) % правая сторона
\qbezier(50,120)(112,138)(175,120) % верх
\qbezier(50,20)(112,2)(175,20) % низ
%
% дырки
%
% левая
\qbezier(15,75)(15,45)(75,40) % низ
\qbezier(20,65)(50,70)(65,42) % верх
%
% правая
%
\qbezier(130,40)(190,45)(200,75) % низ
\qbezier(140,42)(160,85)(196,66) % верх
%
% точки
% группа $P_1$
\put(0,70){\circle*{4}}
\put(0,55){$P_1$}
\put(6,78){\circle*{2}}
\put(15,85){\circle*{2}}
\put(30,95){\circle*{2}}
%
% группа $P_N$
\put(225,70){\circle*{4}}
\put(208,78){$P_N$}
\put(221,62){\circle*{2}}
\put(212,52){\circle*{2}}
\put(195,41){\circle*{2}}
%
% дивизор Тюрина
%
\put(80,70){\circle*{2}}
\put(86,78){\circle*{2}}
\put(95,85){\circle*{2}}
\put(110,95){\circle*{4}}
\put(130,105){\circle*{2}}
\put(150,112){\circle*{2}}
%
%параметры Тюрина
%
\thinlines
\put(110,95){\line(0,1){35}}% древко флага
\put(110,130){\line(2,-1){20}}
\put(110,110){\line(2,1){20}}
\put(114,89){$\ga_s$}
%%%%%%%%%%%%%%%%%%%%%%%%%%%%%%%%%
\put(91,112){$F_s$}
\put(84,64){$\ga_1$}
\put(80,70){\line(0,1){35}}% древко флага
\put(80,105){\line(2,-1){20}}
\put(80,85){\line(2,1){20}}
\put(154,106){$\ga_K$}
\end{picture}
\caption{Base curve with a given flag configuration}\label{}
\end{figure}
%%%%%%%%%%%%%%%%%%%%%%%%%%%%

Given a flag configuration, we define the corresponding \emph{Lax operator algebra}.

For every $s$ we define a filtration $\{0\}\subseteq\g_{s,-k_s}\subseteq\ldots\subseteq\g_{s,i} \subseteq\g_{s,i+1}\subseteq\ldots\g_{s,k_s}=\g$ in $\g$ given by $\g_{s,i}F_{s,j}\subseteq F_{s,\, i+j}$ for all $j$.

\begin{definition}\label{D:LOA}
The linear space $\L$ defined by
\begin{itemize}
\item[$1^\circ$] $\L\subset Mer(\Sigma\to\g)\cap H^0(\Sigma\setminus(\Pi\cup\Gamma))$
\item[$2^\circ$]  $\forall L\in\L$, $\forall s=1,\ldots, K$, the Laurent expansion of $L$ in the neighborhood of $\ga_s$ has the form
    \[
        L_s(z)=\sum\limits_{i=-k_s}^\infty L_{s,i}z^i
    \]
where $L_{s,i}\in\g_{s,i}$ $\forall i$, $z$ is a local coordinate in the neighborhood of $\ga_s$,
\end{itemize}
supplied with the structure of point-wise commutator is called a Lax operator algebra.
\end{definition}
Thus we obtained a \emph{sheaf of Lax operator algebras} over the space of flag configurations. We denote it by $\L$ also.

Let $D$ be a non-negative divisor supported at $\Pi$: $D=m_1P_1+\ldots+m_NP_N$ where $m_1,\ldots,m_N\ge 0$. Let $\L^D=\{ L\in\L\ |\ (L)+\sum\limits_{s=1}^K k_s\ga_s+D\ge 0 \}$. Obviously, $\L^D$ is a finite rank subsheaf in $\L$ (observe that, by definition, $k_s\ge 0$). A point of the total space of the sheaf is given by a pair $\{\{\ga_s,F_s\},L\}$ where $\{\ga_s,F_s\}$ is a flag configuration representing a point of the base, $L$ is a meromorphic function  satisfying the requirements in \refD{LOA}, representing a point of the fibre over $\{\ga_s,F_s\}$. The group $G$ operates on the total space of the sheaf $\L^D$ as follows: $g\{\{\ga_s,F_s\},L\}=\{\{\ga_s,gF_s\}, gLg^{-1}\}$. In particular, $G$ operates on the base by rotation of flags leaving the corresponding points to be fixed.
\begin{definition}\label{D:PhSp}
The quotient of the total space of the sheaf $\L^D$ by the just defined $G$-action serves as a \emph{phase space} of the dynamical system we are going to define. Let this space be denoted by $\P^D$. Thus $\P^D=\L^D/G$.
\end{definition}

The dynamics on the phase space will be given by means of a Lax equation. It is our next step to introduce it.

By \emph{Lax operator} we mean the map of the total space of the sheaf $\L^D$ to $Mer(\Sigma\to\gl (V))$ given by $\{\{\ga_s,F_s\},L\}\to L$, i.e. the map forgetting the base component of a point in $\L^D$. Taking account of $G$-equivariance of that map, we regard to $L$ as to a function on the phase space. In abuse of notation, we will denote the Lax operator by~$L$.

We introduce the sheaf $\M$ similarly to $\L$ with the following distinction. In the \refD{LOA}, we replace $\L$ with $\M$, $L$ with $M$, and the item $2^\circ$ with the following requirement:
\[
   M_s(z)=\frac{\nu h_s}{z}+\sum_{i=-k}^\infty M_{s,i}z^i
\]
where $M_{s,i}\in\g_{s,i}$ for $i<0$, $M_{s,i}\in\g$ for $i\ge 0$, $\nu\in\C$,  $h_s$ is the unique semi-simple element leaving invariant the filtration $\{\g_{s,i}   \}$ and such that ${\ad}\,h_s|_{\g_{s,i}/\g_{s,i+1}}=i\cdot{\rm id}$. Further on, we introduce the sheaf $\M^D$, and the forgetting map on it similarly to $\L^D$. By \emph{$M$-operator} we call the forgetting map on $\M^D$, and denote it by $M$.

The equation
\begin{equation}
 {\dot L} = [L,M]
\end{equation}
where ${\dot L}= dL/dt$, is called the \emph{Lax equation}. This is a system of ordinary differential equations for the dynamical variables including $h$ and the parameters of the principal parts of the meromorphic functions $L$ and $M$ at the points in $\Pi$ and $\Gamma$.

For this system to be closed, it is necessary to give M as a function of L, see \cite{Shein_UMN2016,Sh_DGr}.

%%%%%%%%%%%%%%%%%%%%%%%%%%%%%%%%%%%%%%%%%%%%%%%%%%%%%%%%
\section{Conformal field theory related to a Lax integrable system}
\label{S:CFT}

By \emph{conformal field theory} (CFT) we mean a family of Riemann surfaces, a finite rank bundle (of {\it conformal blocks}) on this family, and a projectively flat connection on this bundle (c.f. \cite{FrSh}). Conventionally, a moduli space of Riemann surfaces with marked points is considered as such family. The principal idea is that they invent a conformal structure on the Riemann surface to make use of advantages of complex algebraic geometry and calculus, and a flat connection on the space of all conformal structures to make the results independent of the choice of any of them. Instead, we consider here the family of spectral curves (figure \ref{Spectral} below) over the phase space $\P^D$  of the integrable system  defined in the previous section.

For every $L\in\P^D$ the curve $\Sigma_L$ given by the equation $\det(L(z)-\l)=0$ is called a {\it spectral curve} of $L$. It is a $n$-fold branch covering of $\Sigma$ where $n=\dim V$.

Thus we have obtained a family of Riemann surfaces over $\P^D$. We may identify this family with $\P^D$ itself.

The conformal field theory we are going to associate with a given integrable system is not of general type. It will be a Wess--Zumino--Novikov--Witten model, which means that there are two (sheaves of) algebras called \emph{gauge} and \emph{conformal} algebras, resp., a  sheaf of vacuum modules over them, the  sheaf of conformal blocks is a certain quotient of the last, and the connection is of {\it Knizhnik--Zamolodchikov} type (c.f. \cite{rTUY,Schlich_DGr,Sh_DGr}).

We choose the algebra $\A_L$ of \emph{Krichever--Novikov functions} on~$\Sigma_L$ as a \emph{gauge algebra} at a point $L$. Let $\Pi_L$ denote the full preimage of $\Pi$ on~$\Sigma_L$. $\A_L$ is defined as the algebra of all meromorphic functions on~$\Sigma_L$, holomorphic except at the points in $\Pi_L$. We regard $\A_L$ as an associative algebra, as well as a commutative Lie algebra.

As a \emph{conformal algebra}, we choose the Lie algebra $\V_L$ of \emph{Krichever--Novikov vector fields} on~$\Sigma_L$. By that, we mean meromorphic vector fields, holomorphic except at the points in $\Pi_L$. We refer to \cite{Schlich_DGr,Sh_DGr} for a detailed presentation of Krichever--Novikov algebras. With every representation of $\Pi_L$ as a disjoint union of two subsets $\rm{In}_L$ and $\rm{Out}_L$ one can associate a structure of \emph{almost graded algebras} (Lie algebras, resp.) on $\A_L$ and $\V_L$, i.e. represent each of them as a direct sum of subspaces
\[
   \A_L=\bigoplus_{i=-\infty}^\infty \A_{L,i},\quad \V_L=\bigoplus_{j=-\infty}^\infty \V_{L,j}
\]
where $\dim\A_{L,i}<\infty$, $\dim\V_{L,j}<\infty$ for all $i,j$, and there exist $M,M'\in\Z_+$ (independent of $i$, $j$) such that
\[
  \A_{L,i}\A_{L,j}\subseteq\sum\limits_{m=i+j}^{i+j+M}\A_{L,m} , \quad [\V_{L,i},\V_{L,j}]\subseteq\sum\limits_{m=i+j}^{i+j+M'}\V_{L,m}.
\]
This enables us to speak of \emph{vacuum modules} (representations) of those algebras. By that, we mean that the module is generated by an element (\emph{the vacuum}) annihilated by any of the subspaces $\A_{L,i}$, $i>0$ ($\V_{L,j}$, $j>0$, resp.).

Given a splitting $\Pi_L={\rm In}_L\cup{\rm Out}_L$, we also can define a bases in $\A$ and $\V$, called \emph{Krichever--Novikov bases}.  A basis element in $\A$ ($\V$, resp.) is given by a point $P\in{\rm In}_L$ and an integer $i$ called a \emph{degree} of the element; we denote it by $A_{L,P,i}$ ($V_{L,P,i}$, resp.). Up to a normalization, the $A_{L,P,i}$ ($V_{L,P,i}$, resp.) is determined by its orders at the points in $\Pi$ (see details in \cite{KNFb,Schlich_DGr,Sh_DGr}). The basis elements of a given degree $i$ and all $P\in{\rm In}_L$ form a base in $\A_{L,i}$ ($\V_{L,i}$, resp.). More generally, one can define the Krichever--Novikov base in the space of $\l$-forms on $\Sigma_L$ for any half-integer $\l$ as well.

From now on we assume that $\Sigma_L$ admits a holomorphic involution $\s_L$, $\s_L\Pi_L=\Pi_L$, and $\s_L\rm{In}_L=\rm{Out}_L$. In particular, the numbers of elements in $\rm{In}_L$ and $\rm{Out}_L$ are equal. The almost graded structures possessing the last property are especially natural. We refer to \cite{Schlich_DGr} for their description.

An involution on the spectral curve comes either from the underlined Lie algebra or from the base curve. The examples of the first type include, among others, Hitchin systems with gauge groups $SO(2n)$, $SO(2n+1)$, $Sp(2n)$. In this case the involution permutes the branches of the spectral curve corresponding to the characteristic roots of opposite signs. The examples of the second type include Hitchin systems with the gauge group $GL(n)$, on a hyperelliptic curve, with $\Pi$ invariant with respect to the hyperelliptic involution.

Given a splitting $\Pi_L={\rm In}_L\cup{\rm Out}_L$, by the \emph{dual splitting} we mean the splitting where the sets ${\rm In}_L$ and ${\rm Out}_L$ are inverted. The corresponding gradings are called \emph{inverse gradings} \cite{Schlich_DGr}. The gauge and conformal algebras  corresponding to the inverse grading are denoted by $\A_L^*$ and $\V_L^*$, resp. The involution $\s_L$ induces two isomorphisms of the almost graded algebras (Lie algebras, resp.): $\s_L:\,\A_L\to\A_L^*$ and $\s_L:\,\V_L\to\V_L^*$.

As \emph{a vacuum module} at the point $L\in\mathcal P$ we choose a semi-infinite wedge space $\F_L=\wedge^{\infty/2}F_L$ for an almost graded module $F_L$ over both $\A_L$ and $\V_L$ \cite{FFx,KaRa,KNFb,Schlich_DGr}. For the last we have a certain selection. The cases described in greater detail are the modules of $\l$-forms on a Riemann surface (on the spectral curve in our case) \cite{KNFb,Schlich_DGr}, and, more generally, modules of sections of holomorphic bundles on the Riemann surface \cite{Sh_DGr}. Together with a pair consisting of an almost-graded module $F_L$ and its semi-infinite wedge space $\F_L$, we will always consider the dual pair of  modules consisting of the contragredient module $F_L^*$ to the module $F_L$, and its semi-infinite wedge space $\F_L^*$.\label{contra} For instance, for the module of $\l$-forms with the almost-grading given by a certain splitting, the dual module is the space of
$(1-\l)$-forms with the almost-grading given by the dual splitting. The Krichever--Novikov bases for modules can be defined as well. For the modules of $\l$- and $(1-\l)$-forms they turn out to be dual with respect to the pairing given by $(\w_\l,\w_{1-\l})=\sum_{P\in{\rm In}}\res_P\w_\l\cdot \w_{1-\l}$. The sum over $\rm In$ can be replaced with the sum over $\rm Out$ with simultaneous change of sign. The pairing between the corresponding semi-infinite degree modules will be defined in \refS{Unit}.

Next we briefly outline the \emph{Sugawara construction} enabling us to turn any projective vacuum $\A$-module to the corresponding $\V$-module. Here we use a "commutative"\ version of the Sugawara construction (see \cite{KNFb, WZWN2,Sh_DGr,Schlich_DGr} for details).

Let $\{ A_{L,P,i}\ |\ P\in{\rm In}_L,\ i\in\Z\}$ be the Krichever--Novikov base in $\A_L$, $\{ \w_{L,P}^j\ |\ P\in{\rm In}_L,\ j\in\Z\  \}$ be the dual base of 1-forms on $\Sigma_L$ (also assumed to be meromorphic and holomorphic outside $\Pi_L$) with respect to the paring $\langle A,\w\rangle = \sum\limits_{Q\in{\rm Out}_L}\res_Q A\w $. Given a vacuum representation $\F$ of the algebra $\A$, let $u(A)$ be the representation operator of an element $A\in\A$. Define the \emph{energy--momentum tensor} $E$ as follows:
\[
  E={1/2}\sum_{i,j=-\infty}^\infty\,\,\sum_{P,P'\in{\rm In_L}}\nord{u(A_{L,P,i})u(A_{L,P',j})}\w_{L,P}^i\w_{L,P'}^j
\]
where the colon brackets mean a normal ordering ($\nord{u(A_{L,P,i})u(A_{L,P',j})}\,=u(A_{L,P,i})u(A_{L,P',j})$ for $i\le j$, and $\nord{u(A_{L,P,i})u(A_{L,P',j})}\,=u(A_{L,P',j})u(A_{L,P,i})$ for $i> j$, except for a finite number of negative values of $i,j$). $E$ is an operator-valued quadratic differential on $\Sigma$. For every $v\in\V$ we define
\[
  T(v)=\nu\cdot\sum_{Q\in{\rm Out}_L}\res_Q(E\cdot v),
\]
where $\nu$ is a normalization constant.
\begin{theorem}
The map $T:\, v\to T(v)$ is a projective representation of $\V_L$:
\[
       T([v,v'])=[T(v),T(v')]+\eta'(v,v')\cdot{\rm id},
\]
($\eta'$ being a 2-cocycle on $\V_L$) related to the representation $u$ of $\A_L$ as follows: for any $v\in\V_L$, $A\in\A_L$
\[
     [T(v),u(A)]=u(\partial_v A).
\]
\end{theorem}
$T$ is called the \emph{Sugawara representation}.

%%%%%%%%%%%%%%%%%%%%%%%%%%%%%%%%%%%%%%%
\section{Conformal blocks and projectively flat connection}\label{S:aff_KN}

Let $X$ be a tangent vector to $\P^D$ at a point $L$. Consider a deformation of the complex structure of the corresponding spectral curve in the direction of $X$. It is given by a Kodaira--Spencer class in $H^1(\Sigma_L, T\Sigma_L)$ where $T\Sigma_L$ is a tangent sheaf on $\Sigma_L$. The cocycle $\rho_L(X)$ representing this class is defined as follows. We consider the family of gluing functions (giving a complex smooth structure on spectral curves) defined in a union of annuli with the "centers" at the points in ${\rm Out}_L$. We denote such gluing function by $d_L$. Regarding $d_L$ as a diffeomorphism of the union of annuli, we set $\rho_L(X)=d_L^{-1}\partial_X d_L$. Apparently, $\rho_L(X)$ is a local vector field in the union. The situation near one of the points $Q^L_j\in{\rm Out}_L$ is illustrated at Figure \ref{Spectral}.

\begin{figure}[h]
\begin{picture}(280,200)
\unitlength=0.8pt
%\put(0,0){\circle*{1}}
% база
\qbezier(70,120)(105,130)(142,133) % верх левая сторона
\qbezier(167,134)(230,140)(350,140) % верх правая сторона
\qbezier(70,120)(20,100)(10,40)
\qbezier(10,40)(90,100)(250,0)
\qbezier(250,0)(250,100)(350,140)
\put(280,40){{\Large $\P^D$}}
\put(230,110){\circle*{3}}
\put(220,103){$L$}
\put(230,110){\line(0,1){50}}
% правая риманова поверхность (папа)
\put(230,195){\oval(40,70)}
\qbezier(220,170)(230,160)(240,180) % нижняя губа
\qbezier(225,167)(230,183)(235,174) % верхняя губа
\qbezier(215,200)(220,216)(225,207) % левая бровь
\qbezier(234,206)(239,215)(244,199)% правая бровь
\qbezier(218,205)(220,200)(222,210) % левое нижнее веко
\qbezier(236,208)(239,200)(242,201)% правое нижн веко
\put(185,195){ $\Sigma_L$}
% отмеченная точка и кольцо
\put(239,223){\circle*{3}}
\put(240,230){$Q^L_j$}
\qbezier(235,230)(230,213)(247,222)% внутр контур
\qbezier(230,230)(225,208)(249,217)% внешн контур
% укрупнение
%
\qbezier(265,250)(320,270)(330,235)%
\put(305,250){\circle*{3}}
\qbezier(295,255)(290,240)(318,248)% внутр контур
\qbezier(285,255)(270,230)(326,240)% внешн контур
\put(270,205){ $d_L$}
%
%векторное поле Кодаиры-Спенсера
%
\put(290,253){\vector(0,-3){10}}
\put(305,240){\vector(4,1){13}}
\put(310,212){$\rho(X)$}
\put(323,225){\line(-1,4){4}}
%
% диффеоморфизм
\put(283,210){\line(1,2){15}}
\put(270,210){\line(-3,1){25}}
%
% векторное поле X
%
\put(230,110){\vector(3,1){40}}
\put(270,120){{ $X$}}
\put(235,118){\vector(3,1){30}}
\put(235,105){\vector(3,1){30}}
% мама
%
\put(80,110){\circle*{3}}
\put(80,110){\line(0,1){50}}
\put(80,195){\oval(40,70)}
\qbezier(70,180)(80,160)(90,170) % нижняя губа
\qbezier(75,174)(80,183)(85,167) % верхняя губа
\qbezier(95,200)(90,216)(85,207) % левая бровь
\qbezier(76,206)(71,215)(66,199)% правая бровь
\qbezier(92,205)(90,200)(88,210) % левое нижнее веко
\qbezier(74,208)(71,200)(68,201)% правое нижн веко
%
% детка
%
\put(155,90){\circle*{3}}
\put(155,90){\line(0,1){25}}
\put(155,133){\oval(20,35)}
\qbezier(148,130)(155,110)(163,130) % нижняя губа
\qbezier(152,125)(155,130)(159,125) % верхняя губа
\qbezier(147,137)(150,143)(154,138) % левая бровь
\qbezier(157,138)(160,143)(164,137)% правая бровь
\qbezier(149,137)(150,133)(152,138) % левое нижнее веко
\qbezier(159,138)(160,133)(162,137)% правое нижн веко
\end{picture}
\caption{Family of spectral curves over the phase
space}\label{Spectral}
\end{figure}
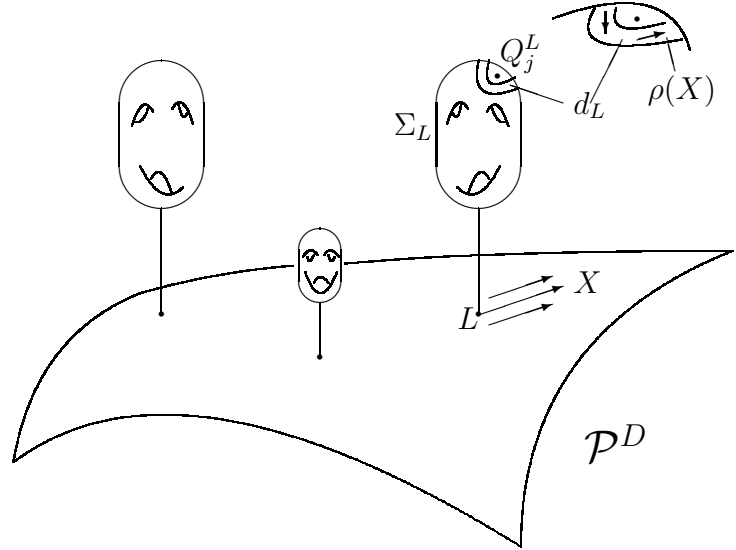

It is shown in \cite{WZWN2,Sh_DGr} that the Kodaira--Spencer cocycle can be represented by a Krichever--Novikov vector field on $\Sigma_L$, and the ambiguity in that representation is compensated by passing to the sheaf of conformal blocks (to be defined below). According to that, below we consider $\rho_L(X)$ as an element of the space $\V_L^{\rm reg}\backslash\V_L/\V_L^+$ where $\V_L^+$ is the direct sum of homogeneous subspaces in $\V_L$ of nonnegative degrees, and $\V_L^{\rm reg}\subset\V_L$ is the subspace of vector fields vanishing at the points in ${\rm Out}_L$. Both these subspaces are Lie subalgebras in $\V_L$.

Let $\F_L$ be a vacuum module of the type discussed in the previous section.

As soon as $\rho_L(X)$ is a Krichever--Novikov vector field, we consider  $T(\rho_L(X))$ where $T$ is the Sugawara representation, and then define the operators
\[  \nabla_{X}=\partial_{X}+T(\rho_L(X)).
\]

Next, we consider the sheaf of $\A_L$-modules $\F_L$ on $\P^D$. Let $\A_L^{reg}\subset\A_L$ be the subalgebra of functions regular at the points in ${\rm Out}_L$. The sheaf $\mathcal C$ of quotient spaces $\F_L/\A_L^{reg}\F_L$ on $\P^D$ is called the sheaf of \emph{coinvariants}, or the sheaf of \emph{conformal blocks} in another terminology. Together with the sheaf $\mathcal C$ we will consider the sheaf ${\mathcal C}^*$ of \emph{dual conformal blocks} which corresponds to the dual module $\F^*_L$ and the algebra $\A^{*reg}_L$. Thus, by definition
\[
     \mathcal C=\F_L/\A_L^{reg}\F_L,\quad \mathcal C^*=\F^*_L/\A_L^{*reg}\F^*_L.
\]
\begin{theorem}[\cite{WZWN2,Sh_DGr,Schlich_DGr}]\label{T:pr_flat} The operators $\nabla_X$ define a projective flat connection $\nabla$ on the sheaf of coinvariants, in particular
\[
 [\nabla_X,\nabla_Y]=\nabla_{[X,Y]}+\l(X,Y)\cdot {\rm id}
\]
where $\l$ is a cocycle on the Lie algebra of tangent vector fields on $\P^D$, ${\rm id}$ is the identity operator.
\end{theorem}
In \cite{WZWN2,Schlich_DGr,Sh_DGr} \refT{pr_flat} is formulated and proved for the conformal field theory on the moduli space of curves with marked points and fixed, up to a certain order, jets at those points. We claim that here (as well as in \cite{Sh_DGr}), the situation is completely similar, and the same proofs are working. Due to this analogy, we call the projective flat connection defined by \refT{pr_flat}, the {\it Knizhnik--Zamolodchikov connection}.

The horizontal sections of the Knizhnik--Zamolodchikov connection are also called \emph{conformal blocks}.

%%%%%%%%%%%%%%%%%%%%%%%%%%%%%%%%%%%%%%%
\section{Quantization of commuting Hamiltonians}

Now we are going to introduce a symplectic
structure on $\P^D$ following the lines of \cite{Klax,Sh_DGr,Shein_UMN2016}.

Let $\Psi$ be the $GL(V)$-valued function on $\Sigma$ diagonalizing $L$:
\[
   \Psi L=K\Psi
\]
where $K$ is the diagonal matrix formed by the eigenvalues of $L$.
$\Psi$ is defined modulo permutations of its rows. Let $\d L$,
$\d\Psi$ and $\d K$ denote the external differentials of $L$, $\Psi$ and $K$, respectively, considered as functions on $\P^D$.

Let $\Omega$ be a  2-form on $\P^D$ taking values in
the space of meromorphic functions on $\Sigma$, defined by the
relation
\[
   \Omega=\tr(\d\Psi\wedge\d L\cdot\Psi^{-1}-\d
   K\wedge\d\Psi\cdot\Psi^{-1}).
\]
$\Omega$ does not depend on the order of the eigenvalues, hence it
is well-defined.

Fix a holomorphic differential $\varpi$ on $\Sigma$ and define a
scalar-valued 2-form $\w$ on $\L^D$ by the relation
\[
    \w=-\frac{1}{2}\left( \sum\limits_{s=1}^K \res_{\gamma_s}\Omega \varpi+
    \sum\limits_{P\in\Pi_L}\res_{P}\Omega \varpi\right).
\]
There is another representation for $\Omega$:
\begin{equation}\label{E:KPh}
 \Omega=2\d\,\tr\left( \d\Psi\cdot\Psi^{-1}K\right)
\end{equation}
which implies that $\w$ is apparently closed.
\begin{theorem}[\cite{Klax}]
$\w$ is a symplectic form on $\P^D$.
\end{theorem}
\noindent We call it {\it Krichever--Phong symplectic form} \cite{Klax,KrPhong}.

Possessing a symplectic structure on $\P^D$, for any $f\in C^\infty(\P^D)$ we can consider the corresponding Hamiltonian vector field $X_f$. The
correspondence $f\to X_f$ is a homomorphism of the
Poisson algebra of classical observables of the Lax integrable
system to the Lie algebra of Hamiltonian vector fields on $\P^D$.

By \refT{pr_flat}, $X\to\nabla_X$ is a projective representation of
the Lie algebra of all smooth vector fields on $\P^D$ in the space of
sections of the sheaf of covariants. Denote this representation by
$\nabla$. The restriction of $\nabla$ to the subalgebra of
Hamiltonian vector fields gives the projective representation of
the last, hence of the Poisson algebra of all classical observables.
\begin{theorem}[\cite{Sh_LKZ,Sh_DGr,Shein_UMN2016}]
If $X$, $Y$ are Hamiltonian vector fields such that their
Hamiltonians Poisson commute then
$[\nabla_X,\nabla_Y]=\l(X,Y)\cdot id$. If the Hamiltonians depend
only on action variables, then $[\nabla_X,\nabla_Y]=0$.
\end{theorem}
\begin{proof} The first assertion immediately follows from
\refT{pr_flat} since $[X,Y]=0$.

Coming to the second assertion, observe that for a Lax equation the spectral curve is an integral of motion
(more precisely, the coefficients of its equation are integrals).
This means that if $X$ is a Hamiltonian vector field such that its
Hamiltonian depends only on the action variables then the complex
structure, hence the transition functions, are invariant along the
phase trajectories of~$X$. Denote by $d_L$ the transition function
over $L$. By its invariance $\partial_Xd_L=0$, hence
$\rho_L(X)=d_L^{-1}\partial_Xd_L=0$, and $\nabla_X=\partial_X$. Let
$H_X$, $H_Y$ be Hamiltonians depending only on the action
variables. Then
$[\nabla_X,\nabla_Y]=[\partial_X,\partial_Y]=\partial_{[X,Y]}$,
and $[X,Y]=0$. This implies the commutativity of $\nabla_X$,
$\nabla_Y$.
\end{proof}

%%%%%%%%%%%%%%%%%%%%%%%%%%%%%%%%%
\section{Unitarity}\label{S:Unit} A goal of the section is to
define an  Hermitian scalar product in the space of sections of the sheaf $\mathcal C$, such that the operators $\nabla_X$ become skew-Hermitian for all Hamiltonian vector fields $X$ on $\P^D$. For this purpose, we first introduce a point-wise pairing of covariants, and then integrate it over the phase space $\P^D$ by an invariant volume form.

Over every point  $L\in\P^D$ we introduce an Hermitian pairing between $\F_L$ and $\F_L^*$ by setting
\begin{equation*}
   \left( f_{i_1}\wedge\ldots\wedge f_{i_k}\wedge f_{i_k+1}\ldots \ |\
   f^*_{j_1}\wedge\ldots\wedge f^*_{j_k}\wedge f^*_{j_k+1}\ldots\right)_L := \d_{i_1j_1}\ldots\d_{i_nj_n}
\end{equation*}
where $\{ f_i\}$, $\{ f_j^*\}$ are dual Krichever--Novikov bases in $F_L$ and $F_L^*$, respectively, $n$ is the number of the exterior cofactor after which  both semi-infinite monomials stabilize, $\d_{ij}$  is the Kronecker symbol. We refer to the end of \refS{CFT} for definitions and notation. In abuse of notation, let a linear operator $\s_L:\, \F_L\to\F_L^*$ be defined by
\[
     \s_L:\, f_{i_1}\wedge\ldots\wedge f_{i_k}\wedge f_{i_k+1}\ldots \to f_{i_1}^*\wedge\ldots\wedge f_{i_k}^*\wedge f_{i_k+1}^*\ldots\ .
\]
Then
\begin{equation}\label{E:pw}
  \langle \phi_1,\phi_2\rangle_L=(\phi_1|\s_L \phi_2)_L,\ \phi_1,\phi_2\in\F_L
\end{equation}
gives a non-degenerate Hermitian bilinear form on $\F_L$.

For the reason $\{f_j\}$ and $\{f_j^*\}$ are dual Krichever--Novikov bases in the  modules  $F_L$ and $F^*_L$, respectively, the scalar product \refE{pw} is the same as introduced  in \cite{KNFb} for an arbitrary genus and the sets ${\rm In}$, ${\rm Out}$ consisting of one point each (the \emph{two-point  case}). A generalization for the \emph{multipoint case}, as well as for another appropriate paring defined in \cite{KNFc}, is given in \cite[Section 7.5]{Schlich_DGr}. The genus zero, two-point case has been considered in \cite{KaRa}.

In the two-point case, for an arbitrary genus, it is stated in \cite{KNFb} that the Sugawara operators $T(v)$ defined in \refS{CFT} are Hermitian with respect to \refE{pw} (a similar statement for a different pairing is formulated in \cite{KNFc}). There is no reason why it could be different in the multipoint case. However, we will formulate this as a
\begin{conjecture}
The Sugawara operators $T(v)$ defined in \refS{CFT} are Hermitian with respect to \refE{pw}.
\end{conjecture}
In course of the proof of the \refT{unit} below we prove, in particular,
that $\langle\cdot,\cdot\rangle_L$ gives a well-defined point-wise scalar product on
coinvariants.

Given a pair of sections $s_1$, $s_2$ of the sheaf $\mathcal C$ define their scalar product by
\begin{equation}\label{E:inv_pair}
  \langle s_1,s_2\rangle=\int_{\P^D}\langle s_1 ,s_2\rangle_L\frac{\w^p}{p!},\ p=\frac{1}{2}\dim\P^D.
\end{equation}
Let $\L^2({\mathcal C},\w^p/p!)$ be the space of quadratically
integrable sections of the  sheaf $\mathcal C$ with respect to the symplectic volume on $\P^D$.
\begin{theorem}\label{T:unit}
For every Hamiltonian vector field $X$ on $\P^D$ the operator $\nabla_X$ in the subspace of smooth sections  in $\L^2({\mathcal C},\w^p/p!)$ is skew-Hermitian with respect to the scalar product \refE{inv_pair}.
\end{theorem}
\begin{remark}
We state this modulo the above conjecture.
\end{remark}
\begin{proof}
First we will prove that the point-wise scalar product on
$\F_L$ is well-defined on coinvariants. We rely on the fact that $\F_L$ is a module of semi-infinite wedge forms, and the different quasihomogeneous almost-graded components of the such module are orthogonal. The point-wise
coinvariants are defined as $\F_L/\A_L^{reg}\F_L$. The degrees of
mon\-omials occuring in the subspace $\A_L^{reg}\F_L$ (including
summands in linear combinations) are obviously less then those of
monom\-ials which form coinvariants. Hence these spaces of
mono\-mials are orthogonal. For this reason the induced scalar product on the
quotient does not depend on the choice of representatives in the
equivalence classes, and we get a well-defined point-wise scalar
product on coinvariants.

Next, we will show that the involution $\s_L$ induces an antiinvolution on Hamiltonian vector fields on $\P^D$. Indeed, the symplectic form is antiinvariant with respect to the involution $\s_L$. For the involutions permuting the branches of the characteristic equation, as discussed in \refS{CFT}, we have $\s_L:\, K\to -K$ in \refE{KPh} which obviously implies $\s_L:\, \Omega\to -\Omega$. For the involutions coming from the hyperelliptic involution on the base the claim is also true for the reason that one of the hyperelliptic coordinates changes its sign \cite{Shein_FAN19}. Besides, the involution is an automorphism of the spectral curve, hence it leaves Hamiltonians invariant. This means that the involution turns a Hamiltonian equation $\dot\xi=\{ H,\xi\}$ to $\dot\xi=-\{ H,\xi\}$ which can be regarded as change of sign of the Hamiltonian, or inversion of time. All in all, the involution takes a Hamiltonian trajectory to the Hamiltonian trajectory with inverse time. On the Hamiltonian vector fields it results in the antiinvolution $v\to -v$.

For a local vector field $X$ on $\P^D$, by $\partial_X$ we mean the
corresponding derivative of local smooth sections of the sheaf~$\mathcal C$. If $X$ is a Hamiltonian vector field then $\partial_X$ is skew-Hermitian with respect to the pairing \refE{inv_pair}. Indeed, by Poincar\'e theorem, the symplectic form and its degrees are absolute integral invariants of Hamiltonian phase flows. Hence $\w^p/p!$ defines an invariant volume form with respect to Hamiltonian phase flows which implies $\partial_X^*=-\partial_X$ (where the upper star denotes the Hermitian conjugation with respect to the pairing \refE{inv_pair}).

By the above conjecture (proven in \cite{KNFb} for the two-point case and arbitrary genus) $T(v)$ is Hermitian with respect to \refE{pw} for all $v\in\V_L$. For $v=\rho(X)$ this reads as $T(\rho(X))^*=T(\s_L(\rho(X)))$. Since for every $L$ both involutions $v\to \s_L(v)$ and $X\to -X$ are induced by the same global automorphism $\s_L$ of the corresponding spectral curve, they commute with the Kodaira--Spencer mapping. For $v=\rho(X)$ this implies  $\s_L(\rho(X))\equiv\rho(-X){\rm mod}\,\V^{*reg}_L$, hence $T(\rho(X))^*=-T(\rho(X))$.

All together,
\[
 (\partial_X+T(\rho(X))^*=-(\partial_{X}+T(\rho(X)),
\]
i.e.
\[
 (\nabla_X)^*=-\nabla_X,
\]
as required.
\end{proof}

%%%%%%%%%%%%%%%%%%%%%%%%%%%%%%%%%%%%%%%%

\end{document}
%%%%%%%%%%%%%%%%%%%%%%%%%%%%%%%%%%%%%%%%%%%%%%%%
%%   THE END
%%%%%%%%%%%%%%%%%%%%%%%%%%%%%%%%%%%%%%%%%%%